\documentclass[12pt,authoryear]{elsarticle}
\usepackage[margin=2.3cm, includehead]{geometry}
\makeatletter
\def\ps@pprintTitle{%
  \let\@oddhead\@empty
  \let\@evenhead\@empty
  \let\@oddfoot\@empty
  \let\@evenfoot\@oddfoot
}
\makeatother
\usepackage{currfile}
\usepackage{mathrsfs}
\usepackage[mathscr]{eucal}
\DeclareMathAlphabet{\mathpzc}{OT1}{pzc}{m}{it}
\usepackage{natbib}
\usepackage{makeidx}
\usepackage{multirow}
\usepackage{multicol}
\usepackage{booktabs}
\usepackage{enumerate}
\usepackage{setspace}
\usepackage{float}
\usepackage{hhline}
\usepackage{bbm}
\usepackage[dvipsnames,svgnames,table]{xcolor}
\usepackage{graphicx}
\usepackage{epstopdf}
\usepackage{ulem}
\usepackage{amsmath}
\usepackage{amsfonts}
\usepackage{url}
\usepackage{amssymb}
\usepackage{graphics}
\usepackage{epsfig}
\usepackage{verbatim}
\usepackage{amsthm}
\usepackage{geometry}
\usepackage{framed}    % to make framed boxes
\usepackage{scalefnt}
\usepackage{alltt}
\usepackage{hyperref}
\hypersetup{
    bookmarks=true,         % show bookmarks bar?
    unicode=false,          % non-Latin characters in Acrobat's bookmarks
    pdftoolbar=true,        % show Acrobat's toolbar?
    pdfmenubar=true,        % show Acrobat's menu?
    pdffitwindow=true,      % page fit to window when opened
    pdftitle={My title},    % title
    pdfauthor={Author},     % author
    pdfsubject={Subject},   % subject of the document
    pdfnewwindow=true,      % links in new window
    pdfkeywords={keywords}, % list of keywords
    colorlinks=true,        % false: boxed links; true: colored links
    linkcolor=blue,         % color of internal links
    citecolor=blue,        % color of links to bibliography
    filecolor=magenta,      % color of file links
    urlcolor=cyan           % color of external links
}

\def\i.i.d.{\buildrel {\rm i.i.d.} \over \sim}

\def\cw#1 { \overset{\mathbb{P}}{\underset{#1}{\longrightarrow}} }

\def\Natu0{\mathbb{N}_0}

\def\E#1{{\mathrm E}\left[#1\right]}

\def\EE#1{{\mathbb E}\left[#1\right]}

\def\Var#1{{\mathrm Var}\left(#1\right)}

\def \rcov#1#2 {{\rm cov}_{#1}\left( #2\right)}

\newtheorem{lemma}{Lemma}

\newtheorem{remark}{Remark}

\newtheorem{model}{Model}

\begin{document}
	\begin{titlepage}	
		\thispagestyle{empty}

\title{
A Classification of Observation-Driven  State-Space Count Models for Panel Data } 

\author{Jae Youn Ahn\fnref{thirdfoot}\corref{bbb}}%\corref{bbb}
\author{Himchan Jeong \fnref{firstfoot}\corref{bbb}}
\author{Yang Lu\fnref{secondfoot}\corref{bbb}}%\corref{bbb}}
\author{Mario V.~W\"uthrich\fnref{fourthfoot}\corref{bbb}}

%\cortext[aaa]{First Authors}
\cortext[bbb]{Corresponding authors/equal contribution}
%\address[1]{Mailing Address 1}
%\address[2]{Mailing Address 2}
%\address[3]{Mailing Address 3}
\fntext[thirdfoot]{Department of Statistics, Ewha Womans University, Seoul, Republic of Korea. Email: \url{jaeyahn@ewha.ac.kr}}
\fntext[firstfoot]{Department of Statistics and Actuarial Science, Simon Fraser University, BC, Canada. Email: \url{himchan_jeong@sfu.ca}}
\fntext[secondfoot]{Department of Mathematics and Statistics, Concordia University, Montreal, QC, Canada. Email: \url{yang.lu@concordia.ca}}
\fntext[fourthfoot]{RiskLab, Department of Mathematics, ETH Zurich, Switzerland.
Email: \url{mario.wuethrich@math.ethz.ch}
}

\begin{abstract}
State-space models are widely used in many applications. In the domain of count data, one such example is the model proposed by \cite{harvey1989time}. Unlike many of its parameter-driven alternatives, this model is observation-driven, leading to closed-form expressions for the predictive density. In this paper, we demonstrate the need to extend the model of \cite{harvey1989time} by showing that their model is not variance stationary. Our extension can accommodate for a wide range of variance processes that are either increasing, decreasing, or stationary, while keeping the tractability of the original model. 
Simulation and numerical studies are included to illustrate the performance of our method.

\end{abstract}

\maketitle




\end{titlepage}

\newpage

\setcounter{footnote}{0}

\section{Introduction}
\label{sec: Introduction}
Time series of counts data are widely used in many areas such as insurance, finance, marketing, economics, etc. According to \cite{cox1981statistical}, there are two major types of time series models, called observation-driven and parameter-driven, respectively. In the count time series framework, the most popular observation-driven time series models are thinning based models, such as INAR($p$) and INGARCH models; see, e.g., \cite{lu2021predictive, davis2021count} for a review. These models are not state-space based, whereas the best known state-space models include, for instance, \cite{zeger1988regression, henderson2003serially, fruhwirth2006auxiliary, cui2009new, davis2009negative, jung2011dynamic, jia2023latent}, to name a few. 

Compared to INAR($p$) and INGARCH models, state-space models have several advantages:
\begin{itemize}
    \item First, it is more convenient to include covariates (i.e., regressors).
    \item  Second, it is more convenient to address missing values, as well as changes of exposures in a state-space framework.
    \item  Third, stationarity (or non-stationary) is more tractable under a state-space framework, and in the case of a stationary process, the marginal distribution is often simple to work out.
\end{itemize}

These three properties can be essential in applications involving panel (i.e., longitudinal) data. One typical example is car insurance pricing, where the insurer observes, for each year, the values of the covariates, as well as a count response variable representing the annual number of claims. Let us explain the importance of these three properties in more detail. 

\paragraph{Allowing for covariates} Most INAR and INARCH (or INGARCH) type models are used without including covariates. The only exceptions we are aware of are  \cite{davis2003observation} and \cite{agosto2016modeling}. These models directly postulate conditional distributions of the future observations, whose parameters are functions of the \textit{current} values of the covariates. A drawback of this approach is that past values of the covariates do not enter into the conditional distribution. To see why past covariate values could be important for car insurance pricing, let us assume, for expository purpose, that the covariate $X_t$ is univariate, and that given past claim numbers $Y_{t-1}, Y_{t-2},\ldots$ and past and current covariate values, the conditional expectation (i.e., the premium) of $Y_t$ is increasing in $X_t$. In other words, $X_t$ measures the underlying risk of the policy. Then, \textit{given $Y_{t-1}, \ldots$} and $X_t$, the premium should be decreasing in past covariate values $X_{t-1},\ldots$, this is because the larger $X_{t-1}, \ldots$, the more ``efforts" the policyholder made in the past to arrive at the given numbers of claims $Y_{t-1},\ldots$. Because these efforts are statistically likely to continue in the future\footnote{And they should be compensated to give incentives for safe driving, due to bonus-malus systems used in insurance pricing.}, the premium should be decreasing in $X_{t-1},\ldots$; we refer to Equation (15) of \cite{DionneVanasse1989} for an example of a premium function that satisfies this decreasingness constraint.

\paragraph{Accounting for missing values and change of exposure} 
In car insurance, it is common for the insurance policies to be analyzed by calendar year. Then for each policy, the first observation is usually left truncated (due to policy inceptions during the calendar year), with an exposure equal to only a fraction of a year. Similarly, the last observation could be censored because of early termination of the policy. These differences of exposure can be conveniently handled by multiplying the stochastic Poisson parameter in the state-space model by an offset term being equal to the fraction of the year covered. It is less straightforward to adjust for such changes of exposures under the INAR and the INGARCH framework, respectively. 

\paragraph{Analyzing stationarity and stationary distribution}
In a longitudinal data context, the likelihood function involves the initial distribution of the observed process. Indeed, for each given individual, the joint probability density function (pdf) of the first $n$ observations can be decomposed as
$$
f(Y_1,Y_2,\ldots,Y_T)=f(Y_1)f(Y_2|Y_1)f(Y_3|Y_1, Y_2)\cdots f(Y_T|Y_{T-1},\ldots, Y_1),
$$
where the first term $f(Y_1)$ is the marginal pdf of the first response $Y_1$. In observation-driven count models such as INAR($p$) and IN(G)ARCH, the subsequent terms corresponding to the conditional distributions are usually more tractable, but the first term can be rather cumbersome. This first term can be omitted, only if the time series length $T$ is very large. 
When $T$ is small, however, its omission induces a bias. 

In state-space models, on the other hand, the term $F(Y_1)$ is often tractable, since $Y_1$ is the output of a latent variable $\Theta_1$, whose distribution is usually chosen simple. For instance, in many parameter-driven models, the latent process $(\Theta_t)$ is assumed stationary and ergodic, following a Gaussian AR(1) or an autoregressive gamma process. Then, the marginal distribution of $\Theta_1$ is simple.

Despite the three aforementioned advantages of state-space count models, many \textit{parameter-driven} state-space models often suffer from a much higher computational burden since the latent process $(\Theta_t)$ has to be integrated out, which leads to a $T$-dimensional integral that may be approximated via Monte Carlo simulation;  see, e.g., \cite{chan1995monte}, \cite{fruhwirth2006auxiliary}. This makes their implementation challenging, especially in a longitudinal data context with a large cross-sectional dimension;  see, e.g., \cite{lu2018dynamic} for a discussion in the context of insurance pricing.  

\cite{harvey1989time}'s model (henceforth HF model) is one of the rare examples of  \textit{observation-driven} state-space models, that enjoys the three properties above, while still being tractable. 
More precisely, in contrast to its \textit{parameter-driven} counterparts, the dynamics of this latter model is not exogenous, but endogenous, in the sense that the dynamics depends not only on the past values of the state variable, but also on the past values of the observed responses
\begin{equation}\label{data-driven}
  \Theta_{t+1}|(\Theta_{1:t}, Y_{1:t})~=~  \Theta_{t+1}|(\Theta_t,Y_{1:t}).
\end{equation}
where $Y_{1:t}=(Y_1,\ldots, Y_t)$ and $\Theta_{1:t}=(\Theta_1,\ldots, \Theta_t)$ denote the processes of the past claim observations, $(Y_t)$, and the latent risk factors up to time $t$, $(\Theta_t)$, respectively. The tractability of the predictive distribution arises from the Poisson-gamma conjugacy, by \textit{assuming} that the conditional distributional on the right hand side of \eqref{data-driven} is a gamma distribution, while the conditional distribution in the measurement equation of $Y_t$, given $\Theta_t$, is assumed to be Poisson. This model can be regarded as the count valued analog of a Bayesian state-space model that relies on conjugate priors, such as \cite{smith1986non} and \cite{shephard1994local} for  real-valued univariate processes, and \cite{uhlig1997bayesian} for real-valued multivariate processes. It has recently been applied by \cite{youn2022simple} in an insurance pricing context. 

In this paper we start by explaining that in the HF model, the state variable follows a (multiplicative) random walk, and it has a non-stationary (increasing) variance process. This explosion-in-variance property might not be appropriate in many applications. Therefore, we extend the HF model to accommodate for various other types of variance dynamics. In particular, we classify the extended class of observation-driven state-space models into several groups:

\begin{itemize}
\item[(a)] The original HF model, which has an \textit{explosive (non-stationary)} dynamics with increasing variance $\Var{\Theta_t}$ with time $t \ge 1$.
\item[(b)] The second class corresponds to the case where, when $t$ goes to infinity, the latent process $(\Theta_{t})$
degenerates to a constant and, hence, the uncertainty related to the non-observability of the latent factor $\Theta_t$ asymptotically vanishes.
\item[(c)] In the third class, the latent process $(\Theta_{t})$ considered in \eqref{data-driven} has a variance process that is bounded (from zero and infinity). This class includes the special cases where the variance process is time-invariant or converging. This third case is probably the most realistic situation for car insurance pricing, where we learn the unobservable risk factors of the insurance policyholders over time, but there always remains some uncertainty. 
\end{itemize}

The models in these three classes differ in terms of their (variance) stationarity properties, but they all enjoy the three aforementioned properties in terms of covariates, change of exposure, and stationary distributions. 

Our paper also contributes to the forecasting literature on exponential moving average or exponential smoothing;  see \cite{hyndman2008forecasting}, Chapter 16, for a review of such methods for count data. This literature has traditionally focused on models with increasing variance, such as the HF model, which gives rise to an exponentially weighted moving average (EWMA) predictor. In this paper we show that many count process models with bounded variance also allow for EWMA predictors, hence, broadening the scope of exponential smoothing methods.  

The rest of the paper is organized as follows. Section \ref{sec.2} extends the HF model, by allowing some of the parameters of the HF model to be time-varying. Section \ref{sec.3} discusses various specifications of this extended HF family, and it classifies them into different classes according to their stationarity (or non-stationarity) behavior, see Table \ref{tab:typology}, below. Section \ref{sec.4} illustrates the difference of their long-term dynamics through simulations. Section \ref{sec.5} compares these models using a real insurance dataset. Section \ref{sec.6} concludes. The mathematical proofs are provided in the appendix.  

%%%%%%%%%%%%%%%%%%%%%%%%%%%%%%%%%%%%%%%%%%%%%%%%%%%%%%%%%%%%%%%%%%%%%%%%%%%%%%%%%%%%%%%%%%%%%

\section{The extended HF model}
\label{sec.2}
Throughout this paper, we use the following notation.
\begin{itemize}
  \setlength\itemsep{0em}
	\item ${\rm Gamma}(\alpha, \beta)$: gamma distribution with shape parameter $\alpha>0$  and
	rate parameter $\beta>0$. It has mean $\alpha/\beta$ and variance $\alpha/\beta^2$. By convention, we use ${\rm Gamma}(0, \beta)$ to denote a constant zero.
	\item ${\rm Beta}(\alpha, \beta)$: beta distribution on $(0,1)$ with mean $\frac{\alpha}{\alpha+\beta}$ and variance $\frac{\alpha\beta}{(\alpha+\beta)^2 (\alpha+\beta+1)}$.
By convention, we use ${\rm Beta}(\alpha, 0)$ to denote a constant one.
	\item ${\rm Pois}( \lambda)$: Poisson distribution
 with mean $\lambda>0$.
 \item ${\rm NB}(\lambda, \Gamma)$: negative binomial (NB) distribution
	with mean $\lambda$ and variance  $\lambda + \lambda^2/\Gamma.$
\end{itemize}
We recall the usual Poisson-gamma relationship. If $Y$ is Poisson,  given $\Theta$, with mean $\Theta$, and $\Theta$ follows a gamma prior distribution, then the posterior distribution of $\Theta$, given $Y$, is still a gamma distribution, and the marginal distribution of $Y$ is a NB distribution.

\subsection{The model}
We provide an observation-driven state-space model with constant mean, which generalizes the HF model.
\begin{model}
\label{mod.0}
Given exogenous processes
$  \left(\lambda_t\right)_{t\ge 1}, \ \left(q_t^*\right)_{t\ge 1}
  \ \hbox{and}\
  \left(q_t^{**}\right)_{t\ge 1}
  $
satisfying for $t\ge 1$
\begin{equation}\label{eq.21}
0\le q_t^{*}\le q_t^{**}\le 1 \quad \text{ and } \quad q_t^{**}, \lambda_t>0,
\end{equation}
the response variables  $(Y_{t})_{t \geq 1}$ and the state-space variables (random effects) $(\Theta_{t})_{t \geq 1}$ satisfy:
\begin{enumerate}[$i)$]
\item The conditional distribution of 
 $Y_t$, given the state variable and past observations\footnote{By convention, for $t=1$, the information set $\sigma(Y_{1:(t-1)})$ reduces to the trivial $\sigma$-field.} $Y_{1:(t-1)}$, is Poisson 
\begin{equation}
\label{eq2}
   Y_t | \left( Y_{1:(t-1)}, \Theta_{1:t} \right) \sim {\rm Pois}\left( \lambda_t\Theta_t\right), \qquad \text{ for $t \ge 1$}.
\end{equation}
  \item   At time $t=1$, $\Theta_1$ is gamma distributed as
\begin{equation}\label{eq1}
  \Theta_1\sim {\rm Gamma}\left(\alpha_{1|0}, \beta_{1|0}\right),
\end{equation}
where, for identification purposes, we assume equal deterministic parameters $\alpha_{1|0}=\beta_{1|0}>0$, so that $\EE{\Theta_1}=1$. 
 \item  At time $t\ge 1$, the filtering distribution of $\Theta_{t}$, given past observations 
 $Y_{1:t}$, is gamma
 \begin{equation}\label{eq.3400}
 \Theta_{t}| Y_{1:t}\sim {\rm Gamma}\left(   \alpha_{t}, \beta_{t}\right),
 \end{equation}
where $\alpha_{t}>0$ and  $\beta_{t}>0$ are deterministic functions of
$Y_{1:t}$ and $\lambda_{1:t}$ up to time $t$
\begin{equation}\label{eq41}
	\alpha_{t}=
\begin{cases}
		\alpha_{t|t-1}+Y_t, & \hbox{if $Y_t$ is observed};\\
		\alpha_{t|t-1}, & \hbox{otherwise};
\end{cases}
\end{equation}
and
\begin{equation}\label{eq42}
	\beta_{t}=
\begin{cases}
		\beta_{t|t-1}+\lambda_t,
& \hbox{if $Y_t$ is observed} ;\\
		\beta_{t|t-1}, & \hbox{otherwise}.
\end{cases}
\end{equation}
\item At time $t+1\ge 1$, the predictive distribution of $\Theta_{t+1}$, given $Y_{1:t}$, is  gamma with
\begin{equation}
	\label{stillgamma0}
\Theta_{t+1}|Y_{1:t} \sim {\rm Gamma}\left( \alpha_{t+1|t}, \beta_{t+1|t}\right),
\end{equation}
with for $t\ge 1$
\begin{eqnarray*}
  \alpha_{t+1|t}&=& q_{t}^*\alpha_{t}+ \left( q_{t}^{**}- q_{t}^{*}\right)\beta_{t},\\
  \beta_{t+1|t}&=&q_{t}^{**}\beta_{t}.
\end{eqnarray*}
\end{enumerate}
\end{model}
Definitions \eqref{eq41} and \eqref{eq42} follow from Bayes' rule, and $\alpha_t>0$ and $\beta_t>0$ are deterministic functions of the past observations $Y_{1:t}$, up to time $t$, and of $\lambda_{1:t}$, the latter allows to
integrate time-varying covariates.
From this, we deduce that the conditional distribution of $Y_{t}$, given $Y_{1:(t-1)}$, is negative binomial with 
\begin{equation}\label{NB log-likelihood sequence}
Y_{t}\,|\, Y_{1:(t-1)} \sim {\rm NB}
\left(
\lambda_{t}\, \frac{\alpha_{t|t-1}}{\beta_{t|t-1}},\,
\alpha_{t|t-1}\right),
\end{equation}
where the conditional probability mass function is given as follows
\begin{equation}\label{NB condi}
  f\left(\left. Y_t \right| Y_{1:(t-1)}\right)
  = \frac{\Gamma\left(Y_t+\alpha_{t|t-1}\right)}{Y_t!\,\Gamma\left( \alpha_{t|t-1}\right)}
  \left(\frac{\lambda_t }{\lambda_t  + \beta_{t|t-1}}\right)^{Y_t}
       \left(\frac{\beta_{t|t-1}}{\lambda_t  + \beta_{t|t-1}}\right)^{\alpha_{t|t-1}}.
     \end{equation}
In particular, the predictive mean is
\begin{align}
    \label{predictivemean}
    \mathbb{E} [ Y_{t}|Y_{1:(t-1)}]&= \lambda_t \, \frac{\alpha_{t|t-1}}{\beta_{t|t-1}}. 
\end{align}

\subsection{Stochastic representation of the observation-driven property}
Model \ref{mod.0} defines the conditional distributions of the latent variable $\Theta_t$, given the past (or past and current) observations. However, it does not provide a state equation directly linking $\Theta_{t}$ with $\Theta_{t+1}$. To work out this state equation, we first recall the following lemma. 
\begin{lemma}[\cite{lukacs1955characterization}]\label{lem.1}
Consider two independent random variables
\begin{equation}
\label{betaBeta}
\begin{aligned}
 \Theta&\sim {\rm Gamma}(\alpha, \beta),\\
  B&\sim {\rm Beta}\left(q^*\alpha, (1-q^*)\alpha\right),%\\
\end{aligned}
\end{equation}
where $\alpha>0, \beta>0$, and $q^*\in (0,1]$ are given constants.
Then, their product $$\Theta B\sim {\rm Gamma}(q^*\alpha, \beta).$$
\end{lemma}
As a consequence, if $\eta$ is independent of $\Theta$ and $B$, and $\eta \sim {\rm Gamma} ((q^{**}-q^*)\beta, q^{**}\beta)$ with constant $q^{**}$ such that $q^{**}\geq q^*$, then we have:
\begin{equation}\label{lem.eq.1}
 \frac{\Theta B}{q^{**}} +\eta \sim {\rm Gamma}\left( q^*\alpha+\left(q^{**}-q^{*} \right)\beta, q^{**}\beta\right).
\end{equation}
Formula \eqref{lem.eq.1} implies the following stochastic representation of the latent process $(\Theta_t)_{t\ge 1}$: 
\begin{equation}\label{parametric state space update}
\Theta_{t+1} = \frac{\Theta_{t} B_{t+1}}{q_{t}^{**}} +\eta_{t+1},
\end{equation}
where $$B_{t+1}\,|\,(Y_{1:t}, \Theta_{1:t})\sim {\rm Beta}(q_{t}^*\alpha_{t}, (1-q_{t}^*)\alpha_{t}),$$ and $$\eta_{t+1}\,|\,(Y_{1:t}, \Theta_{1:t})\sim {\rm Gamma}(
  (q_t^{**}-q_t^*)\beta_t, q_{t}^{**}\beta_{t}).$$ Moreover, $B_{t+1}$ and $\eta_{t+1}$ are conditionally independent, given $Y_{1:t}$ and $\Theta_{1:t}$.
  The observation-driven nature of the evolution in \eqref{data-driven} is evident from the evolution mechanism in
  \eqref{lem.eq.1}-\eqref{parametric state space update}, and this justifies the
  choice of \eqref{stillgamma0} by an explicit example.

\subsection{Link with other time series models}
\paragraph{Link with random coefficient AR(1) processes} We remark that conditional on the information up to time $t$ in \eqref{parametric state space update}, we have
\[
\EE {\frac{B_{t+1}}{q_{t}^{**}}\,\bigg\vert\, \Theta_{1:t}, Y_{1:t}}=\frac{q_{t}^*}{q_{t}^{**}} \leq 1.
\]
 Thus, the process $(\Theta_t)$ can be compared with a (random coefficient) auto-regressive process, see, e.g., \cite{joe1996time} and \cite{jorgensen1998stationary}, in which the first term in \eqref{parametric state space update} describes a (stochastic) thinning of the previous state $\Theta_{t}$, and $\eta_{t+1}$ adds new noise to the update. 
 Our specification of $(\Theta_t)$ differs from this random coefficient literature by the fact that we consider an endogenous, i.e., observation-driven dynamics of $(\Theta_t)$.

 \paragraph{Link with Kalman filters} It is well known in a linear Gaussian state-space model\footnote{That is, the conditional joint distribution of the pair $(Y_{t+1}, \Theta_{t+1})$ given past information $Y_{1:t}$ and $\Theta_{1:t}$ is Gaussian.} that all the conditional/filtering/predictive distributions are Gaussian; see \cite{durbin2012time}, Chapter 4. This result is based on the Gaussian-Gaussian conjugacy, as well as the closure of the Gaussian distribution to convolution and scaling. Because our model is based on the Poisson-gamma conjugacy, as well as the closure of the gamma distribution to scaling and convolution,\footnote{This holds for fixed scale parameter.} it can be viewed as a count variable analogue of the Kalman filter. More generally, state-space models based on other conjugate priors have been proposed by \cite{smith1986non}, \cite{shephard1994local}, and \cite{uhlig1997bayesian}, to name but a few.

%%%%%%%%%%%%%%%%%%%%%%%%%%%%%%%%%%%%%%%%%%%%%%%%%%%%%%%%%%%%%%%%%%%%%%%%%%%%%%%%%%%%%%%%%%%%%

\section{A classification according to the behavior of the variance process}\label{sec.3}

In this section, we show that Model \ref{mod.0} can allow for various forms of variance processes,
e.g.,  of increasing, decreasing, constant or stationary type.

\subsection{The static shared random effect model}\label{sec.3.2}
The model with shared (or static) random effect assumes that given a time-invariant latent variable $\Theta_t\equiv \Theta$, a.s., for all $t\ge 1$, and with $\Theta$ following a $\text{Gamma}(\beta_{1|0}, \beta_{1|0})$ distribution, the counts $Y_t$ are conditionally independent with a $\text{Pois}(\lambda_t \Theta)$ distribution.
For $\lambda_t\equiv\lambda>0$, this is a special case of the \cite{BuhlmannStraub} credibility model.

This is a static shared random effect model is obtained from  Model \ref{mod.0}
by setting
\begin{equation}\label{eq.22}q_{t}^*=q_{t}^{**}=1 \quad \Longleftrightarrow \quad  B_{t+1}\equiv 1, \ \eta_{t+1} \equiv 0.\end{equation}

\subsection{HF model: A model with increasing variance}\label{sec.3.3}
To analyze more sophisticated cases, let us first investigate some properties
of Model \ref{mod.0} concerning the first and second order moments of the processes. 

\begin{lemma}\label{lem.app.1}
  In Model \ref{mod.0} we have the following moment behaviors for $t\ge 1$
  \begin{enumerate}[$i)$]
    \item {$\EE{\alpha_{t}}=\beta_{t}$, $\EE{\Theta_{t}}=1$ and $\EE{Y_{t}}=\lambda_{t}$.}
       \item $\EE{\Var{\Theta_{t} \,|\, Y_{1:t}}}=\frac{1}{\beta_{t}}$.
    \item $\EE{\Var{\Theta_{t+1} \,|\, Y_{1:t}}}=\frac{1}{q_{t}^{**}}
    \frac{1}{\beta_{t}}$.
    \item  $\Var{\EE{\Theta_{t+1} \,|\, Y_{1:t}}}=
        \left( \frac{q_{t}^{*}}{q_{t}^{**}} \right)^2
    \Var{\EE{\Theta_{t} \,|\, Y_{1:t}}}$.
  \end{enumerate}
\end{lemma}
\begin{proof}
See \ref{proofoflemma1}.
\end{proof}
Property $i)$ says that process $(\Theta_t)$ is mean-stationary. Properties $iii)$ and $iv)$ allow to analyze the variance behavior of $(\Theta_t)$. Indeed, by the total variance decomposition formula, we get
\begin{equation}
\label{nextvariance}
\begin{aligned}
\Var{\Theta_{t+1}}&=\EE{\Var{\Theta_{t+1} \,|\, Y_{1:t}}}+\Var{\EE{\Theta_{t+1} \,|\, Y_{1:t}}}\\
&=\frac{1}{q_{t}^{**}}
   \frac{1}{\beta_{t}}+\left( \frac{q_{t}^{*}}{q_{t}^{**}} \right)^2
    \Var{\EE{\Theta_{t} \,|\, Y_{1:t}}}.
\end{aligned}
\end{equation}
On the other hand, by using the total variance decomposition formula again, we have
\begin{align}
\label{previousvar3iance}
  \Var{\Theta_{t}}&=  \EE{\Var{\Theta_{t} \,|\, Y_{1:t}}}+
    \Var{\EE{\Theta_{t} \,|\, Y_{1:t}}}\nonumber \\
    &=\frac{1}{\beta_{t}}+ \Var{\EE{\Theta_{t} \,|\, Y_{1:t}}}.
\end{align}
Let us now compare equations  \eqref{nextvariance} and \eqref{previousvar3iance}. Because $\frac{1}{q_{t}^{**}} \geq 1$ the first term in
 \eqref{nextvariance} is larger than the first one  of \eqref{previousvar3iance}. Similarly, because $\left( \frac{q_{t}^{*}}{q_{t}^{**}} \right)^2 \leq 1$, the second term in \eqref{nextvariance} is  smaller than the second one of \eqref{previousvar3iance}.
As a result, the variance process $(\Var{\Theta_t})$ in our model is not necessarily monotone.  
Throughout the rest of this section, we will discuss cases in which this sequence of variances $(\Var{\Theta_t})$ is either increasing, time-varying, or decreasing.

\medskip

In this subsection we start by describing an increasing case.
The HF model, \cite{harvey1989time}, is obtained from Model \ref{mod.0} under the extra constraints for $t\ge 1$
\begin{equation}\label{eq.14}
q_{t}^*=q_{t}^{**}=q ~ \in~ (0,1).
\end{equation}
\cite{harvey1989time}'s original formulation is without covariates, \cite{gamerman2013non} extend their model by introducing time-varying exogenous covariates $(\lambda_t)$.

In this model, the stochastic representation \eqref{parametric state space update} becomes
\begin{equation}
\label{martingale}
\Theta_{t+1} = \frac{\Theta_{t} B_{t+1}}{q} ,
\end{equation}
which implies that $\E { \Theta_{t+1}  | Y_{1:t}}=\Theta_{t}$. In other words, $(\Theta_t)$ is a martingale with respect to the filtration generated by $(Y_{1:t})$.

The following lemma shows that under some conditions, this martingale has an explosive variance behavior. 
\begin{lemma}[Explosive variance in the HF model]
\label{lem.5}
Under the HF Model, if the exogenous process $(\lambda_{t})$ is bounded both both above and below across time $t\ge 1$, then $(\Var{\Theta_t})$ and $(\Var{Y_t})$  increase to infinity when $t$ goes to infinity. 
\end{lemma}
\begin{proof}
See \ref{proofoflem.5}.
\end{proof}

One advantage of the HF model is that under this model, the predictive mean \eqref{predictivemean} is an exponentially weighted moving average (EWMA) of past observations. The literature on exponential moving average forecasting of counts has traditionally focused on \textit{non-stationary} models only;  see, e.g., \cite{hyndman2008forecasting}, Chapter 17. It is seen in the next subsection, however, that it is not the only possible specification leading to EWMA predictors. 

\subsection{A model with converging variance}\label{sec.3.4}
In this subsection, we discuss a special case of Model \ref{mod.0} which is asymptotically strongly stationary. By \textit{strongly} stationarity, we mean that the conditional distribution of $Y_{t+h}$, given $Y_t$, converges to a non-degenerate distribution. In other words, in the long run, the process $(Y_t)$ evolves in a ``steady state", i.e., an ``equilibrium". This implies,  in particular, that the variance and mean of the process converge to positive constants.\footnote{However, the strong stationarity is stronger than the second order convergence of the variance and mean of the process. Indeed, the stationarity implies the convergence of any moments, whenever they exist.} Note also that stationarity does not mean ``time-invariant". Indeed, for large time horizons $h$, the variance of $Y_{t+h}$ converges to a positive constant.

Specifically, we assume, in Model \ref{mod.0}, that $q_t^*, q_t^{**}$ and $\lambda_t$ are all time-invariant
for $t\ge 1$
  \begin{equation}
  \label{asymptoticallystationary}
q_t^*=pq, \qquad q_t^{**}=q,
\end{equation}
for $p, q \in(0,1)$, and
\begin{equation}
    \label{lambdaalsoconstant}
 \lambda_t=\lambda. 
\end{equation}
Then, we have for $t\ge 1$
\begin{equation} \label{param converging}
  \begin{aligned}
  \alpha_t&=\alpha_{t|t-1}+Y_t=
  p q\alpha_{t-1}+q\left( 1-p\right)\beta_{t-1}+Y_t, \\
  \beta_t&=\beta_{t|t-1}+\lambda_t=
           q\beta_{t-1}+\lambda_t,
           \end{aligned}
\end{equation}
where the second identities need $t\ge 2$.
Because our model has as predictive distribution a NB distribution, with the number of trial parameter satisfying a linear recursion, this model coincides with the NB-INGARCH(1,1) model, see 
\citet{gonccalves2015infinitely}, who established its stationarity.
\begin{lemma}[\citet{gonccalves2015infinitely}]
\label{lemma zzz}  
If in Model \ref{mod.0} the parameters satisfy \eqref{asymptoticallystationary} and \eqref{lambdaalsoconstant}, 
then the process $(Y_t)$ is asymptotically strongly stationary. In particular, $\left(\Var{\Theta_t}\right)$ and $\left(\Var{Y_t}\right)$ converge.
\end{lemma}

Under assumptions \eqref{asymptoticallystationary} and \eqref{lambdaalsoconstant}, the predictive mean \eqref{predictivemean} is once again an EWMA of the past observations. In other words, it is an extension of the standard EWMA forecasting literature; see  \cite{hyndman2008forecasting}. 

\begin{remark}\normalfont
There are also other stationary count process models allowing for EWMA predictive mean formulas, such as the INGARCH(1,1), see \cite{ferland2006integer},  or the NB-INGARCH, see \cite{zhu2011negative}. These other models, however, do not admit a state-space representation and, therefore, do not possess the three properties mentioned in the introduction. In other words, among all the stationary count models with an EWMA predictive mean, the model of \cite{gonccalves2015infinitely} has the advantage of having a state-space representation.
\end{remark}

\subsection{A model with decreasing variance}\label{sec.4.1}
In this subsection, we consider the special case of Model \ref{mod.0} with the constraint for $t\ge 1$
\begin{equation}
\label{condition2}
q_t^*=p ~\in [0,1) \quad \hbox{and}\quad q_t^{**}=1.
\end{equation}
This implies for $t\ge 2$, see \eqref{param converging},
$$
\alpha_t=p \alpha_{t-1}+ (1-p)\beta_{t-1}+Y_t \quad \hbox{and}\quad \beta_t=\beta_{t-1}+\lambda_t.
$$
%in case $Y_t$ is observed.
If $(\lambda_t)$ is bounded from below, then both processes $(\beta_t)$ and $(\alpha_t)$ go to infinity when $t$ increases
to infinity.

Let us study the variance of this process. By comparing \eqref{nextvariance} and \eqref{previousvar3iance} with the condition in \eqref{condition2}, we get
\begin{equation}
\label{varianceisdecreasing}
    \Var{\Theta_{t+1}}-\Var{\Theta_{t}} =\left({p}^2-1\right)\Var{ \EE{\Theta_{t}\,|\, Y_{1:t}}}=\left({p}^2-1\right)\Var{ \frac{\alpha_{t}}{\beta_{t}}} <0.
    \end{equation}
    Thus, the latent process $(\Theta_t)$, and hence $(Y_t)$,  have both a decreasing variance (for the latter it is sufficient to assume that $(\lambda_t)$ is bounded). Moreover, we have the following stronger result.
\begin{lemma}\label{lem.3}
  Under Model 1 and constraint \eqref{condition2}, if the exogenous process $(\lambda_t)$ is bounded both from below and above by positive constants, then $(\Var{\Theta_t})$ converges to zero, when $t$ goes to infinity, and $(\Var{Y_t})$ converges to
$\lambda$, under the additional assumption \eqref{lambdaalsoconstant}.
\end{lemma}

\begin{proof}
See \ref{proofoflem.3}.
\end{proof}

\subsection{A model with constant variance} \label{sec.4.2}
In this subsection, we discuss a special case of Model \ref{mod.0} for which the variance is a constant.
Note that the model in \citet{gonccalves2015infinitely}, which is the model in Section \ref{sec.3.4}, requires $(\lambda_t)$ to be time-invariant \eqref{lambdaalsoconstant}, which might be too restrictive for insurance applications. Moreover, the variance is not time-invariant in Lemma \ref{lemma zzz}, but it only converges to a constant at infinity. In the following, we look at another  type of stationarity property, by looking only at the variance of the process $(\Theta_t)$.\footnote{Note that the process $(\Theta_t)$ has a constant mean by Lemma \ref{lem.1}.} However, instead of requiring the variance to converge to a constant when time $t$ goes to infinity, we request it to remain constant for any $t$,\footnote{In particular, we do not require the process to be covariance stationary. That is, the covariance function of the process ${\mathrm Cov}(\Theta_t, \Theta_{t-h})$ can depend on $t$.} that is, for $t\ge 1$
   \begin{equation}\label{eq.sta.1}
  1=\EE{\Theta_t}\quad\hbox{and}\quad \frac{1}{\beta_{1|0}}=\Var{\Theta_t}.
  \end{equation}
This will in turn allow us to relax the time-invariance assumption on $(\lambda_t)$. More precisely, by comparing \eqref{nextvariance} and \eqref{previousvar3iance}, we get immediately the following result.
\begin{lemma}
In Model \ref{mod.0}, the variance process $(\Var{\Theta_t})$ is constant, if and only if $q^*_{t}$ and $q^{**}_{t}$ satisfy the following equation for all $t\ge 1$
 \begin{equation}
   \label{constantvariance}
   \frac{1}{ q^{**}_{t}}
 \frac{1}{\beta_t}+ \left(\frac{q^*_{t}}{q^{**}_{t}}\right)^2 \left(\frac{1}{\beta_{1|0}}-\frac{1}{\beta_{t}}\right)=\frac{1}{\beta_{1|0}}.
 \end{equation}
\end{lemma}
Thus, there are infinitely many possible combinations of $q^*_{t}$ and $q^{**}_{t}$ in order for $\Var{\Theta_t}$ to remain constant.
Among such choices, motivated by the conditional linear auto-regressive structure (CLAR(1)) in \cite{ar1}, we may assume the following updating rule
  \begin{equation}\label{eq.up.1}
  \EE{\Theta_{t+1}\,\vert\, Y_{1:t}}=p\EE{\Theta_{t}\,\vert\, Y_{1:t}}+1-p, 
  \end{equation} for some $p\in(0,1)$, which is equivalent to the condition
 \begin{equation}
 \label{eq.24}
  \frac{q^*_{t}}{q^{**}_{t}}=p.
  \end{equation}
  Then, under this additional assumption \eqref{eq.up.1}, requirement \eqref{constantvariance} becomes
  for $t\ge 1$
 \begin{equation}\label{eq.up.3}
q_{t}^{**}=
\frac{\beta_{1|0}}{ p^2\beta_{1|0}+\left(1-p^2\right)\beta_t},
\end{equation}
which can be calculated recursively from $\beta_t=q^{**}_{t-1}\beta_{t-1}+\lambda_t$, $t \ge 2$,
with initialization $\beta_1=\beta_{1|0}+\lambda_1$.

\subsection{A model with bounded variance}
For some applications, the assumption of a time-invariant $(\lambda_t)$ imposed in the previous subsection might be too restrictive. In this subsection, we relax this assumption, and consider a model satisfying \eqref{asymptoticallystationary} only, but not \eqref{lambdaalsoconstant}. Then, we get the following result.
\begin{lemma}
\label{lmm5bis}
    If in Model 1, \eqref{asymptoticallystationary} is satisfied, and if the process $(\lambda_t)$ is bounded from both above and below by positive constants, then the variance process $(\Var{Y_t})$ is bounded from above. 
\end{lemma}
\begin{proof}
See \ref{proofoflem5bis}.
\end{proof}

\subsection{A typology of models according to the variance process}
To summarize, the following table lists all the different models considered in this section. 
\begin{table}[H] \centering
\begin{tabular}{|c|c|c|c|c|c|}\hline
Model & Condition \\ \hline\hline
 %   i.i.d. random effect & $q_{t}^*=0, \ q_{t}^{**}=1$ \\ \hline
    Shared random effect  & $q_{t}^*=q_{t}^{**}=1$ \\ \hline
  HF model with increasing (explosive) variance &  $q_{t}^*=q_{t}^{**}=q$, $(\lambda_t)$ bounded \\ \hline
            Decreasing variance & $q_t^{*}=p, \ q_t^{**}=1$, $(\lambda_t)$ bounded\\  \hline
    Converging variance &  $q_t^*=pq, \ q_t^{**}=q,  \ \lambda_t=\lambda$ \\ \hline
    Bounded variance & $q_t^*=pq, \ q_t^{**}=q$, $(\lambda_t)$ bounded \\ \hline
     Constant variance & Eq. \eqref{constantvariance} \\ \hline

\end{tabular}
\caption{Typology of various special cases of Model 1 according to the long-run behavior of the variance process. All the constants $p$ and $q$ lie strictly between 0 and 1 in this table. By ``$(\lambda_t)$ bounded", we mean that it is both upper and lower bounded by positive constants. }
\label{tab:typology}
\end{table}

Because of Equation \eqref{predictivemean}, the models  with a constant $q_{t}^*$ in this table yield a prediction of $Y_{t+1}$ as a exponential moving average of past observations $Y_1,\ldots, Y_t$. In this regard, Model 1 broadens the scope of exponential smoothing based forecasting methods. Indeed, when it comes to count data, their focus was predominantly on the HF model, specifically in the context of count data;  see \cite{hyndman2008forecasting}, Chapter 16.

%%%%%%%%%%%%%%%%%%%%%%%%%%%%%%%%%%%%%%%%%%%%%%%%%%%%%%%%%%%%%%%%%%%%%%%%%%%%%%%%%%%%%%%%%%%%% 

\section{Numerical illustration}
\label{sec.4}
In this section, we simulate trajectories of the various examples considered in Section \ref{sec.3}, and we illustrate the differences in long-term behavior between them. Throughout all examples, we set $\alpha_{1|0}=\beta_{1|0}=3$, and we let $\lambda_t=1$ to simulate 5,000 independent trajectories of $(\Theta_t)$ for $t=1,\ldots, T=50$, under each of the following specifications on the dynamics of $(\Theta_t)$:
\begin{itemize}
    \item Increasing variance of $(\Theta_t)$ (2nd row of Table \ref{tab:typology}): $q_t^*=q_t^{**}=0.8$,
    \item Decreasing variance of $(\Theta_t)$ (3rd row of Table \ref{tab:typology}): $q_t^*=0.8, \ q_t^{**}=1$,
    \item Converging variance %\footnote{Note that this case may not be applicable to the real data analysis as we cannot assure that $\lambda_t$ is constant. Further, this case could be also interpreted as a special case of the bounded variance model.}
       of $(\Theta_t)$ (4th row of Table \ref{tab:typology}): $q_t^*=0.8, \ q_t^{**}=0.9$,
    \item Constant variance of $(\Theta_t)$ (6th row of Table \ref{tab:typology}): $q_t^*=0.9q_t^{**}, \ q_t^{**}=\frac{\alpha_0}{\alpha_0 \cdot 0.9^2 +(1-0.9^2)\beta_t}=\frac{3}{2.43+0.19\beta_t}$.
\end{itemize}

\subsection{Long-run behavior of $(\Theta_t)$}
For each of the four models above, we display, in Figure \ref{Figure 1}, four independent paths over $T=50$ time periods. We observe in the northwest panel that the magnitude of the variation of $(\Theta_t)$ becomes bigger over time, reflecting an increasing variance process. Moreover, the trajectories are highly persistent, which echos the martingale property \eqref{martingale} in the HF model. In the northeast panel, all trajectories of $(\Theta_t)$ tend to vary less and less over time, which is consistent with the decreasing variance specification. Moreover, all the trajectories fluctuate around one positive value, which is expected because of the constant mean property of the process. In the model with constant variance, it is observed that the fluctuation level of $(\Theta_t)$ is stable over time compared to the other scenarios. Lastly, it is shown that the fluctuation level of $(\Theta_t)$ in the converging variance case is between the fluctuation levels in the constant and decreasing cases.

\begin{figure}[H]
    \centering
    \includegraphics[width=1\textwidth]{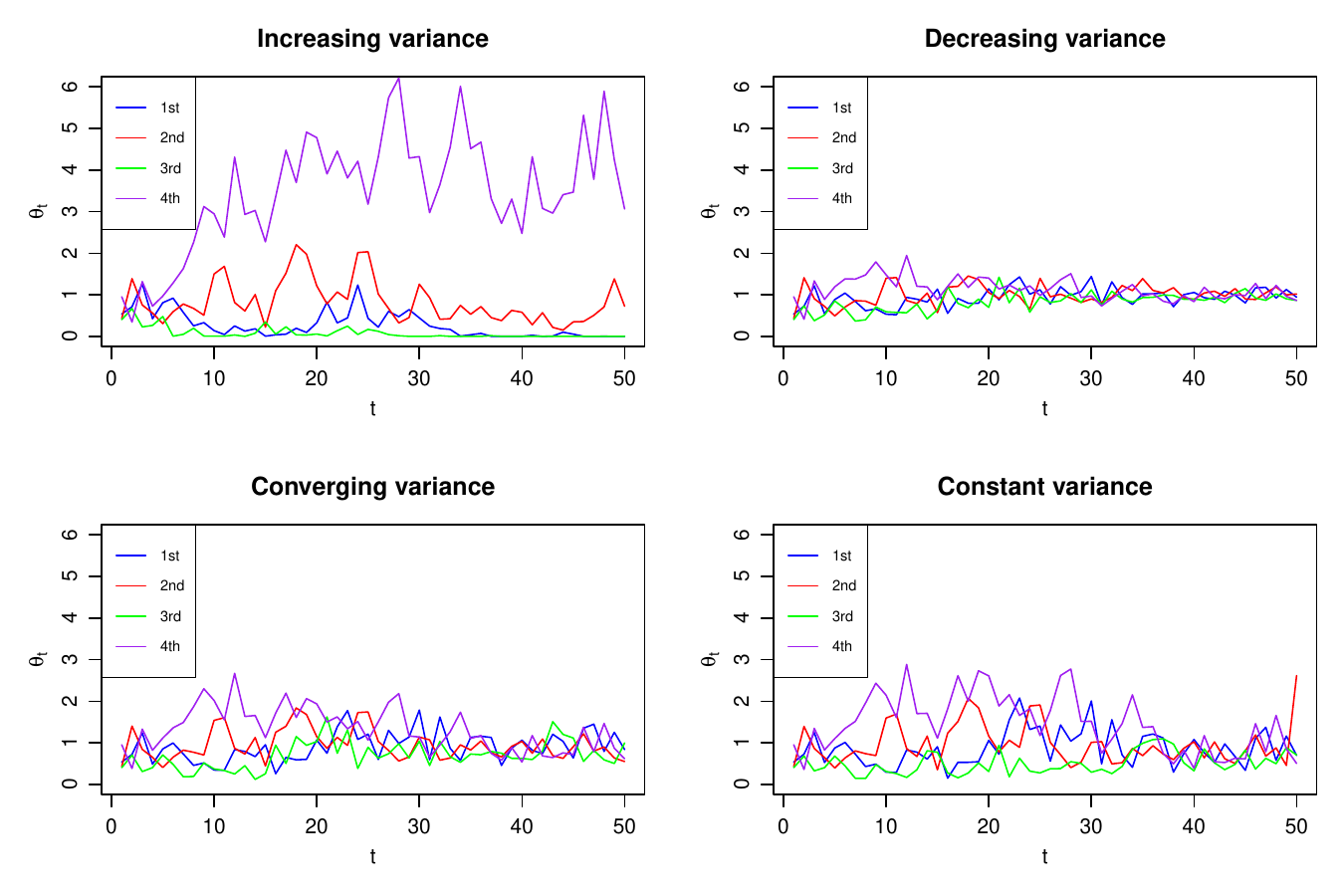}
    	\vspace{-35pt}
     \caption{Four independent trajectories of $(\Theta_t)$
 under each of the four specifications. Northwest panel: the model with increasing variance. Northeast panel: the model with decreasing variance. Southwest panel: the model with converging variance. Southeast panel: the model with constant variance. }
 \label{Figure 1}
\end{figure}

\subsection{Variance of $(\Theta_t)$}
For each of the above four models, we plot the empirical density plots of $\Theta_t$ at different times $t=1,5,20,50$, where each time series $(\Theta_t)$ is simulated 5,000 times. From Figure \ref{Figure 3} we observe the following:
\begin{itemize}
\item In the HF model with increasing variance, the distribution of $\Theta_{50}$ has both a thicker right tail, and a much higher peak near zero compared to the distributions of $\Theta_1$ and $\Theta_5$. This reflects the increasing variance of $(\Theta_t)$ over time under a constant mean, as shown in Figure \ref{Figure 3}.
    \item On the other hand, in the model with decreasing variance, the distribution of $\Theta_{50}$ is much more concentrated around the mean value 1 compared to those of $\Theta_1$ and $\Theta_5$, as the variance of $(\Theta_t)$ decreases over time.
    \item In the model with converging variance, we observe that the distribution of $\Theta_t$ becomes more concentrated as $t$ increases. Moreover, the distribution of $\Theta_{50}$ is significantly different from, say, $\Theta_{20}$.
    \item In the model with constant variance, the distributions of $\Theta_1$, $\Theta_5$, $\Theta_{20}$, and $\Theta_{50}$ are quite close, which reflects the fact they have the same mean and variance.
\end{itemize}

\begin{figure}[H]
    \centering
   \includegraphics[width=0.95\textwidth]{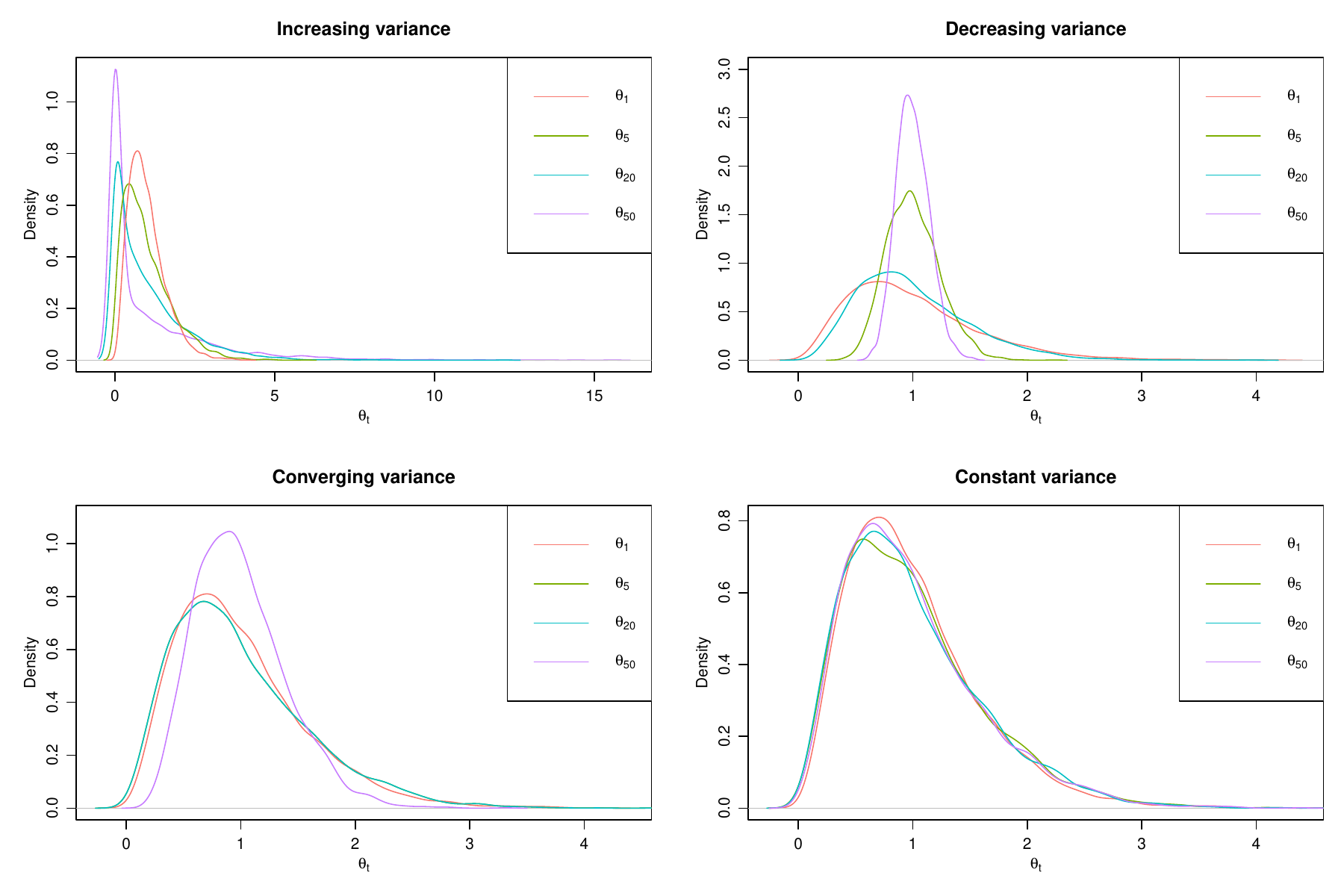}
        	\vspace{-20pt}
     \caption{Empirical density of $\Theta_t$ under each of the four scenarios at times $t=1,5,20,50$}
 \label{Figure 3}
\end{figure}

        	\vspace{-25pt}

%%%%%%%%%%%%%%%%%%%%%%%%%%%%%%%%%%%%%%%%%%%%%%%%%%%%%%%%%%%%%%%%%%%%%%%%%%%%%%%%%%%%%%%%%%%%%
                
\section{Real data analysis}
\label{sec.5}
We use the LGPIF (Local Government Property Insurance Fund) data from the state of Wisconsin. Although the dataset encompasses claims information across multiple types of coverages, in our analysis, we only focus on inland marine (IM) claims. The dataset consists of 6,775 observations from 1,234 policyholders, longitudinally observed for the period of the years 2006-2011.\footnote{The data is publicly available at \url{https://sites.google.com/a/wisc.edu/jed-frees/\#h.lf91xe62gizk}} We use the observations between 2006 and 2010 for model estimation, while the observations from year 2011 are set aside for out-of-sample validation. We refer the reader to \citet{frees2016multivariate} for a  detailed explanation about the data. Table \ref{Table 1} provides a brief summary statistics of the observed policy characteristics. We have one categorical covariate (entity location) available in the dataset with the following values: ``City", ``County", ``Miscellaneous", ``School", ``Town", and ``Village". We code this covariate as 5 binary variables (dummy coding), corresponding to the indicators of ``City", ``County", ``School", ``Town", and ``Village", with ``Miscellaneous" as the reference group. We also have two continuous covariates related to the coverage amount (i.e., the maximal amount covered per claim) and the deductible amount (i.e., the minimal damage to trigger a claim payment). These two covariates may vary in time for a given policyholder. Thus, we cannot fit the model with converging variance, since this latter requires time-invariant covariates (and expected frequencies $\lambda_t$, respectively). 

\begin{table}[h!t!]
\begin{center}
\caption{Policy characteristics used as covariates} \label{tab:datasummary}
\resizebox{!}{3cm}{
\begin{tabular}{l|lrrr}
\hline \hline
Categorical & Description &  & \multicolumn{2}{c}{Proportions} \\
levels \\
\hline
TypeCity & Indicator for city entity:           &  & \multicolumn{2}{c}{14.00 \%} \\
TypeCounty & Indicator for county entity:       &  & \multicolumn{2}{c}{5.78 \%} \\
TypeMisc & Indicator for miscellaneous entity:  &  & \multicolumn{2}{c}{11.04 \%} \\
TypeSchool & Indicator for school entity:       &  & \multicolumn{2}{c}{28.17 \%} \\
TypeTown & Indicator for town entity:           &  & \multicolumn{2}{c}{17.28 \%} \\
TypeVillage & Indicator for village entity:     &  & \multicolumn{2}{c}{23.73 \%} \\
\hline
 Continuous & & Minimum & Mean & Maximum \\
 variables \\
\hline
CoverageIM  & Logged coverage amount of IM claim   &  0 & 0.85
            & 46.75\\
lnDeductIM  & Logged deductible amount for IM claim     &  0 & 5.34
            & 9.21\\
\hline \hline
\end{tabular}}
\label{Table 1}
\end{center}
\end{table}

By letting $i$ be index of the policyholders, $i=1,\ldots, N=1,234$, and letting $T_i$ be the maximal number of observations for the $i^{th}$ policyholder, one can write the full log-likelihood as,
see \eqref{NB log-likelihood sequence},
\begin{equation}\label{fulllik}
\ell = \sum_{i=1}^N \sum_{t=1}^{T_i}\log p(Y_{i,t} | Y_{i,1:(t-1)}), \qquad Y_{i,t} | Y_{i,1:(t-1)} \sim {\rm NB}\left(\lambda_{i,t}
\, \frac{\alpha_{i,t|t-1}}{\beta_{i,t|t-1}},\,
\alpha_{i,t|t-1} \right),
\end{equation}
with expected frequency $\lambda_{i,t} = \exp( \mathbf{x}_{i,t} \eta )$,
regression parameter $\eta \in {\mathbb R}^d$,  and $\mathbf{x}_{i,t}\in {\mathbb R}^d$ are the observable policy characteristics of policyholder $i$ at time $t$ of dimension $d=8$; note that we add lower indices $i$ to all parameters, as these can now
be policyholder dependent. We consider special cases of Model \ref{mod.0},
namely, we assume $q^*_t/q^{**}_t=p \in [0,1]$ for all $t\ge 1$, and, moreover, $q^{**}_t \in (0,1]$ should
not depend on $i$.

This gives us recursive formulas for the shape and rate parameters 
$$\alpha_{i,t+1}=p q_{t}^{**}\alpha_{i,t}+q_{t}^{**}\left( 1-p\right)\beta_{i,t}+Y_{i,t+1} \quad\hbox{and}\quad \beta_{i,t+1}=q_{t}^{**}\beta_{i,t}+\lambda_{i,t+1},$$
for $t\ge 1$, and with initial values $\alpha_{i,1}= \alpha_{1|0} + Y_{i,1}$, $\beta_{i,1} = \beta_{1|0} + \lambda_{i,1}$, and $\beta_{1|0}=\alpha_{1|0}$. Moreover, we have for $t\ge 1$
\begin{equation*}
\alpha_{i,t+1|t} = \alpha_{i,t+1}-Y_{i,t+1} \quad \text{ and } \quad
\beta_{i,t+1|t} = \beta_{i,t+1}-\lambda_{i,t+1},
\end{equation*}
and we initialize all policyholders $i$ as follows $\alpha_{i,1|0}=\alpha_{1|0}$ and
$\beta_{i,1|0}=\beta_{1|0}$.

This allows us to implement the log-likelihood function \eqref{fulllik} for given observations
${\bf Y}=(Y_{1, 1:T_1}, \ldots, Y_{N, 1:T_N})$. Set maximal observation period $T=\max_{1\le i \le N} T_i$.
Then, the log-likelihood $\ell=\ell_{\bf Y}(\vartheta)$ is a function of the parameters
\begin{equation}\label{def: vartheta}
  \vartheta=(\beta_{1|0},  p, q^{**}_{1:(T-1)}, \eta) ~\in~ {\mathbb R}_+ \times [0,1] \times (0,1]^{T-1}\times {\mathbb R}^d.
\end{equation}  
Let us now consider the following models:
\begin{itemize}
   \item \text{Independent} latent factors model: $
\alpha_{i,t} = \alpha_{1|0}, \
\beta_{i,t} = \beta_{1|0}$ for all $t \geq 1$. 
    \item \text{Shared} random effect model: $p=1,\  q_t^{**}=q=1$.
    \item \text{Increasing} variance of $(\Theta_t)$: $p=1$, $q_t^{**}=q \in (0,1)$.
              \item \text{Decreasing} variance of $(\Theta_t)$:  $p\in (0,1), \ q_t^{**}=q=1$.
    \item \text{Constant} variance of $(\Theta_t)$: $p\in (0,1), \ q_t^{**}=\frac{\beta_{1|0}}{p^2\beta_{1|0} +(1-p^2)\beta_t}$.
\end{itemize}
We do not report the bounded variance case since the estimate lies in the boundary
to the increasing variance case.

As all of these models satisfy the generalized linear model (GLM) assumption $\EE{Y_{i,t}} = \lambda_{i,t}=\exp( \mathbf{x}_{i,t} \eta )$, and we use the following two-step estimate approach to estimate $\vartheta$ given in see \eqref{def: vartheta}:
\begin{enumerate}
 \item Estimate regression parameter $\eta \in {\mathbb R}^d$ from the standard NB GLM, which means we do not consider the serial correlations among $Y_{i, 1:T_i}$ at this stage. Note that this approach still yields a consistent estimate of $\eta$ as long as the mean model is correctly specified (but it is still less efficient as the variance structure may be misspecified).
\item After $\eta$ has been estimated from Step 1, estimate the parameters of the random effects dynamics such as $\beta_{1|0}$, $p$ and $q$ (if available).
\end{enumerate}
This two-step approach is consistent by the usual arguments on pseudo likelihood estimation, see \cite{gourieroux1984pseudo}, and it has two advantages. First, the second numerical optimization step is simple since it involves only a smaller number of parameters that are $\beta_{1|0}$, $p$ and $q$. Second, using this approach, we get the same estimator for the regression coefficients $\eta$ in front of the covariates for all the models considered, making it easier to compare them. It would have also been possible to estimate all the parameters together using maximum likelihood estimation, but implementation is more cumbersome and convergence may be an issue.

The model with constant variance shows the best goodness-of-fit as shown in Table \ref{Table 2} so that the constant variance model is the best in terms of AIC and the shared random effect model is the best in terms of BIC, while the difference is small. Note that the parameter estimation with decreasing variance model was unable to find a set of parameters sufficiently different from the shared random effect model. In this regard, one can conclude that the decreasing variance model is not suitable for this database.

\begin{table}[!h]
\caption{Estimated model parameters and goodness-of-fit for the considered models}
\centering
\begin{tabular}[t]{lccccc}
\toprule
\multicolumn{1}{c}{ } & Independent & Shared & Increasing & Decreasing & Constant \\
\midrule
$\beta_{1|0}$ & 0.488 & 0.651 & 0.786 & 0.651 & 0.603\\
$p$ & 0 & 1 & 1 & 1.000 & 0.937\\
$q$ & 1 & 1 & 0.830 & 1 & - \\ \midrule
Loglik & -934.135 & -905.357 & -904.317 & -905.357 & -902.019\\
AIC & 1886.271 & 1828.713 & 1828.633 & 1830.713 & 1824.039\\
BIC & 1946.068 & 1888.511 & 1895.075 & 1897.155 & 1890.481\\
\hline\hline
\end{tabular}
\label{Table 2}
\end{table}

Using the observations from year 2011 as the out-of-sample validation set, we assess the predictive performance of the aforementioned models. We use the RMSE (root mean-squared error), the MAE (mean-absolute error), and the PDL (Poisson deviance loss) defined as follows
$$
\begin{aligned}
\text{RMSE}&= \sqrt{\frac{1}{|\mathcal{T}|}\sum_{i \in \mathcal{T}}\left(Y_{i,T_i+1} - \widehat{Y}_{i,T_i+1}\right)^2}, \\
\text{MAE}&= \frac{1}{|\mathcal{T}|}\sum_{i \in \mathcal{T}}\left|Y_{i,T_i+1} - \widehat{Y}_{i,T_i+1}\right|, \\
\text{PDL}&= \frac{1}{|\mathcal{T}|}\sum_{i \in \mathcal{T}}2\left(\widehat{Y}_{i,T_i+1} - {Y}_{i,T_i+1}- {Y}_{i,T_i+1} \log \left(\frac{\widehat{Y}_{i,T_i+1}}{{Y}_{i,T_i+1}}\right) \right), \\
\end{aligned}
$$
where $\mathcal{T}$ is the number of observations in the validation set $\mathcal{T}$, and
$\widehat{Y}_{i,T_i+1}$ are the forecasts obtained from the fitted models. We prefer a model with lower values of RMSE, MAE, and/or PDL, and it turns out that the model assuming increasing variance shows the best predictive performance in our example, as shown in Table \ref{Table 3}. This change of ranking with respect to Table \ref{Table 2} may have many reasons, e.g., non-stationarity of the data which likely increases the state-space process if not properly modeled. This closes our example.
\begin{table}[!h]
\caption{Out-of-sample validation performance}
\centering
\begin{tabular}{lrrrrr}
\toprule
  & Independent & Shared & Increasing & Decreasing & Constant\\
\midrule
RMSE & 9.0586 & 0.7091 & 0.5896 & 0.7091 & 0.8240\\
MAE & 0.3821 & 0.1107 & 0.1048 & 0.1107 & 0.1143\\
PDL & 0.8407 & 0.2523 & 0.2425     & 0.2523 & 0.2572\\
\hhline{======}
 \label{Table 3}
\end{tabular}
\end{table}

%%%%%%%%%%%%%%%%%%%%%%%%%%%%%%%%%%%%%%%%%%%%%%%%%%%%%%%%%%%%%%%%%%%%%%%%%%%%%%%%%%%%%%%%%%%%%

\section{Conclusion}
\label{sec.6}
In this paper, we expanded the observation-driven state-space model of \cite{harvey1989time} to a broader spectrum of specifications characterized by various variance process behaviors. They are suitable for count processes with a constant mean, but with increasing, decreasing, constant, converging, or bounded variance process. These models inherit most of the major advantages of state-space models, but are more tractable for regression modeling than their parameter-driven counterparts.  Additionally, we elucidated the relationship of this model class with the INGARCH literature, see \cite{gonccalves2015infinitely}, and also drew connections to the forecasting literature that focuses on exponential smoothing, see \cite{hyndman2008forecasting}.

%%%%%%%%%%%%%%%%%%%%%%%%%%%%%%%%%%%%%%%%%%%%%%%%%%%%%%%%%%%%%%%%%%%%%%%%%%%%%%%%%%%%%%%%%%%%%

\section*{Acknowledgments}
Jae Youn Ahn is partly supported by a National Research Foundation of Korea (NRF) grant funded by the Korean Government
%[grant number 2022R1F1A1064048]
and Institute of Information \& communications Technology Planning \& Evaluation (IITP) grant funded by the Korea government (MSIT).
% [grant number RS-2022-00155966].
Yang Lu thanks NSERC through a discovery grant [RGPIN-2021-04144, DGECR-2021-00330]. Himchan Jeong
is supported by the Simon Fraser University New Faculty Start-up Grant (NFSG).

%%%%%%%%%%%%%%%%%%%%%%%%%%%%%%%%%%%%%%%%%%%%%%%%%%%%%%%%%%%%%%%%%%%%%%%%%%%%%%%%%%%%%%%%%%%%%

\bibliographystyle{apalike}
%\bibliography{CTE_Bib_HIX}
\bibliography{CTE_Bib_HIX2}

%%%%%%%%%%%%%%%%%%%%%%%%%%%%%%%%%%%%%%%%%%%%%%%%%%%%%%%%%%%%%%%%%%%%%%%%%%%%%%%%%%%%%%%%%%%%%

\appendix
\section{Proofs}
\subsection{Proof of Lemma \ref{lem.app.1}}
\label{proofoflemma1}
\paragraph{Part $i)$}
We proceed by induction. From \eqref{eq1}, we get $\EE{\Theta_1}=1$ and $\EE{Y_1}=\lambda_1$. Hence $\EE{\alpha_1}=\beta_1$. 
  Assume that
$  \EE{\Theta_{t}}=1$ and $\EE{\alpha_{t}}=\beta_{t}
$
for some $t \ge 1$. First, we have $\EE{\Theta_{t+1}}=1$ from the following consideration, in view of
\eqref{stillgamma0},
  \[
  \begin{aligned}
    \EE{\Theta_{t+1}}&=\EE{\EE{\Theta_{t+1} \,|\, Y_{1:t}}}
    =\EE{\EE{\left.\frac{\alpha_{t+1|t}}{\beta_{t+1|t}} \right| Y_{1:t}}}\\
  &=\EE{\frac{q_{t}^*\alpha_{t}+\left( q_{t}^{**}- q_{t}^{*}\right)\beta_{t}}{q_{t}^{**}\beta_{t}}}\\
  &= \frac{q_{t}^{*}}{q_{t}^{**}}\EE{\frac{\alpha_{t}}{\beta_{t}}}+\frac{q_{t}^{**}-q_{t}^{*}}{q_{t}^{**}}=1.
  \end{aligned}
  \]
 The also  implies $\EE{Y_t}=\lambda_t$ for all $t\ge 1$.
 Then, from the definition in \eqref{eq41} and \eqref{eq42} along with the induction assumption, we have
 for $t\ge 1$
  \[
  \begin{aligned}
  \EE{\alpha_{t+1}}&=\EE{q_{t}^*\alpha_{t}+\left( q_{t}^{**}- q_{t}^{*}\right)\beta_{t}+Y_{t+1}}\\
  &=q_{t}^*\EE{\alpha_{t}}+\left( q_{t}^{**}- q_{t}^{*}\right)\beta_{t}+\EE{Y_{t+1}}\\
  &=q_{t}^{**}\beta_{t}+\lambda_{t+1}=\beta_{t+1}.
  \end{aligned}
  \]

\paragraph{Part $ii)$} This is a direct consequence of Part $i)$.

\paragraph{Part $iii)$} We have
  \[
    \EE{\Var{\Theta_{t+1}\,|\,Y_{1:t}}}=
    \frac{1}{\beta_{t+1|t}}\EE{\EE{\Theta_{t+1}\,|\, Y_{1:t}}}
    =\frac{1}{q_{t}^{**}\beta_{t}},%\EE{\EE{\Theta_{t}\,|\, Y_{1:t}}}
    %=\frac{1}{q_{t}^{**}}\EE{\Var{\Theta_{t}\,|\, Y_{1:t}}},
  \]
  where the first equation follows from \eqref{stillgamma0}.

 \paragraph{Part $iv)$} Equation \eqref{stillgamma0} implies
$
\EE{\Theta_{t+1}\,|\, Y_{1:t}}=\frac{\alpha_{t+1|t}}{\beta_{t+1|t}}=
\frac{\alpha_{t}}{\beta_{t}}\frac{q_{t}^*}{q_{t}^{**}}+
\frac{q_{t}^{**}-q_{t}^*}{q_{t}^{**}},
$
which implies
\[
\Var{\EE{\Theta_{t+1}\,|\, Y_{1:t}}}=\left( \frac{q_{t}^{*}}{q_{t}^{**}} \right)^2\Var{\frac{\alpha_{t}}{\beta_{t}}}
=\left( \frac{q_{t}^{*}}{q_{t}^{**}} \right)^2\Var{\EE{\Theta_{t}\,|\, Y_{1:t}}}.
\]

\subsection{Proof of Lemma \ref{lem.5}}
\label{proofoflem.5}
By comparing equations \eqref{nextvariance} and \eqref{previousvar3iance}, we have
\begin{equation}\label{eq.15}
\Var{\Theta_{t+1}}-\Var{\Theta_{t}}= \left( \frac{1}{q}-1\right) \frac{1}{\beta_{t}}>0.
\end{equation}
If the exogenous process $(\lambda_t)$ is bounded from above by $M>0$, then $\beta_t$, which satisfies the recursion $\beta_t=q \beta_{t-1}+\lambda_t$, is bounded from above by $\frac{M}{1-q}$. Hence $\frac{1}{\beta_{t}}$ is bounded from below. 
Thus, the sequence of variances $(\Var{\Theta_t})$ increases to infinity when $t$ goes to infinity.

Furthermore, if $\lambda_t$ is bounded from below by $m>0$, then by the total variance decomposition formula, $\Var{Y_t}$ increases to infinity as well as $t\to\infty$.

\subsection{Proof of Lemma \ref{lem.3}}
\label{proofoflem.3}
    Because the sequence of variances $(\Var{\Theta_t})$ is decreasing but positive, it converges to a non-negative constant.In other words, the left hand side of equation \eqref{varianceisdecreasing} goes to zero as $t\to\infty$. As a consequence, the right hand side of \eqref{varianceisdecreasing} also goes to zero. Thus, $\Var{ \frac{\alpha_{t}}{\beta_{t}}}= \frac{1}{\beta_{t}^2} \Var{ \alpha_{t}}$ converges to zero.     
    Hence, the second term in  \eqref{previousvar3iance} goes to zero when $t$ goes to infinity. As for the first term in  \eqref{previousvar3iance}, we have $ \EE{ \frac{\alpha_{t}}{\beta^2_{t}}}=  \frac{1}{\beta^2_{t}} \EE{\alpha_{t} }=\frac{1}{\beta^2_{t}} \beta_{t} $ goes to zero as well since $\beta_t$ goes to infinity as $t$ goes to infinity. As a consequence, we deduce that $(\Var{\Theta_t})$ converges to zero when $t$ goes to infinity.

    Finally, by the total variance decomposition
    \begin{equation}\label{variance decomposition Y}
    \Var{ Y_t }=\mathbb{E}\Big[ \Var{ Y_t|\Theta_t}\Big]+  \Var{   \mathbb{E}[Y_t|\Theta_t]}
    =\lambda_t\mathbb{E}[\Theta_t]+ \lambda_t^2 \Var{ \Theta_t}.
   \end{equation}
 Under the additional assumption \eqref{lambdaalsoconstant}, the latter 
    converges to $\lambda \mathbb{E} [\Theta_t]=\lambda \mathbb{E} [\Theta_1]=\lambda$ as $t$ goes to infinity. 
    
\subsection{Proof of Lemma \ref{lmm5bis}}
\label{proofoflem5bis}
First, since $(\lambda_t)$ is bounded, $\beta_t$ is also bounded from above and below by positive constants. Then, by
\eqref{variance decomposition Y} and applying the total variance decomposition formula once more
    \begin{align*}
     \Var{ Y_t}&=\lambda_t\mathbb{E}[\Theta_t]+\lambda_t^2\Big( \Var{ \alpha_t/\beta_t}+\mathbb{E}[\alpha_t/\beta^2_t] \Big).
    \end{align*}
    Thus it suffices to show that $ \Var{\alpha_t}$ is bounded.
    We have
    \begin{align}
        \Var{\alpha_t }&= \Var{pq\alpha_{t-1}+\beta_{t-1}q(1-p)+Y_t} \nonumber \\
       &=  \Var{ pq\alpha_{t-1}+Y_t}\nonumber \\
       &=\mathbb{E}\Big[  \Var{pq\alpha_{t-1}+Y_t|Y_{1:(t-1)} }\Big]+ \text{Var}\Big( \mathbb{E}[pq\alpha_{t-1}+Y_t|Y_{1:(t-1)}]\Big) \nonumber \\
                       &=\mathbb{E}\Big[  \Var{Y_t|Y_{1:(t-1)}}\Big]+\text{Var}\Big( pq\alpha_{t-1}+  \mathbb{E} [Y_t|Y_{1:(t-1)}] \Big)\label{step1}\\
                       &=\mathbb{E}\left[\lambda_t\frac{\alpha_{t|t-1}}{\beta_{t|t-1}}+  \lambda_t^2\frac{\alpha_{t|t-1}}{\beta^2_{t|t-1}}\right]
                         +\text{Var}\left(pq\alpha_{t-1}+\lambda_t\frac{\alpha_{t|t-1}}{\beta_{t|t-1}} \right)
      , \label{step2}
 \\
       &=\left(\frac{\lambda_t}{q \beta_{t-1}}+\frac{\lambda_t^2}{q^2 \beta^2_{t-1}}\right)\mathbb{E}\left[pq \alpha_{t-1}+q(1-p)\beta_{t-1}\right] \nonumber \\
       & \qquad +  \text{Var} \left( pq\alpha_{t-1}+\lambda_t \frac{ pq\alpha_{t-1}+q(1-p)\beta_{t-1}}{q \beta_{t-1}} \right), 
 \nonumber   \end{align}
    where from \eqref{step1} to \eqref{step2} we have used the mean and variance formula of the conditional distribution of $Y_t$, given $Y_{1:(t-1)}$, which is NB, see \eqref{predictivemean}. 

    The first term on the right hand side is upper bounded, since $\mathbb{E}[\alpha_{t-1}]=\beta_{t-1}$ is bounded.
    The last term can be written as, note $\beta_{t-1}$ is deterministic,
    $$
    \text{Var}\Big(\alpha_{t-1} (pq+ \lambda_t p/\beta_{t-1})\Big)=\text{Var}\Big( p  \frac{q \beta_{t-1}+\lambda_{t}}{\beta_{t-1}}\alpha_{t-1} \Big)=\Var {p\frac{\beta_t}{\beta_{t-1}} \alpha_{t-1}}.
    $$
    Thus, by dividing both sides of \eqref{step2} by $\beta_{t}^2$, which is both upper and lower bounded by positive constants, we get
    $$
    \Var {\frac{\alpha_t}{\beta_t}}= p^2    \Var{ \frac{\alpha_{t-1}}{\beta_{t-1}}}+ \text{bounded term}.
    $$
By recurrence, we deduce that $\Var{\frac{\alpha_t}{\beta_t}}$ is upper bounded. Therefore, $\Var{\alpha_t}$ is also upper bounded. Hence, $\Var{Y_t}$ is also upper bounded.

\end{document}